\documentclass[11pt,a4paper]{article}
\usepackage{fullpage}
\usepackage{authblk}
\usepackage{tikz}

\usepackage{ifthen}
\usepackage{graphics,epsfig}
\usepackage{amsmath,amsfonts,amssymb}
\usepackage{algorithm}
\usepackage{algpseudocode}
\usepackage{color}

\newtheorem{theorem}{Theorem}
\newtheorem{lemma}[theorem]{Lemma}

\newtheorem{proposition}[theorem]{\bf Proposition}
\newtheorem{definition}[theorem]{Definition}
\newcommand{\qed}{\hfill $\Box$ \medbreak}
\newenvironment{proof}{\noindent {\bf Proof.}}{\qed}

\newcommand{\cA}{{\cal A}}
\newcommand{\cB}{{\cal B}}
\newcommand{\cC}{{\cal C}}

\newcommand{\cF}{{\cal F}}

\newcommand{\cL}{{\cal L}}
\newcommand{\cR}{{\cal R}}
\newcommand{\cS}{{\cal S}}
\newcommand{\local}{\mbox{{\tiny \sf LOC}}}
\newcommand{\aut}{\mbox{{\tiny\sf AUT}}}
\newcommand{\lcp}{\mbox{{\tiny\sf LCP}}}
\newcommand{\llcp}{\mbox{{\tiny\sf LogLCP}}}
\newcommand{\id}{\mbox{\footnotesize\rm id}}
\newcommand{\AMOS}{\mbox{\sc amos}}
\newcommand{\ALTS}{\mbox{\sc alts}}
\newcommand{\EXTS}{\mbox{\sc exts}}
\newcommand{\ITER}{\mbox{\sc iter}}
\newcommand{\AND}{\mbox{\sc and}}
\newcommand{\OR}{\mbox{\sc or}}
\newcommand{\TREE}{\mbox{\sc tree}}
\newcommand{\DIAM}{\mbox{\sc diam}}
\newcommand{\COVER}{\mbox{\sc cover}}
\newcommand{\MISS}{\mbox{\sc miss}}
\newcommand{\MISSLIFT}{\mbox{\sc miss}^\uparrow}
\newcommand{\COL}{\mbox{\sc prop-col}}

\newcommand{\ID}{\mbox{ID}}
\newcommand{\LD}{\mbox{{\sf LD}}}
\newcommand{\LDO}{\mbox{{\sf LDO}}}
\newcommand{\coLD}{\mbox{\rm co-{\sf LD}}}
\newcommand{\LCL}{\mbox{{\sf LCL}}}
\newcommand{\NLD}{\mbox{{\sf NLD}}}
\newcommand{\NLDO}{\mbox{{\sf NLDO}}}
\newcommand{\coNLD}{\mbox{\rm co-{\sf NLD}}}
\newcommand{\BPLD}{\mbox{{\sf BPLD}}}
\newcommand{\BPNLD}{\mbox{{\sf BPNLD}}}
\newcommand{\all}{\mbox{{\sf All}}}
\newcommand{\ind}{\mbox{{\rm index}}}
\newcommand{\data}{\mbox{{\rm data}}}
\newcommand{\dist}{\mbox{{\rm dist}}}

\begin{document}

\title{Local Distributed Verification\footnote{The first, third and fourth authors received additional supports from the ANR project DISPLEXITY.}}

\author[1,2]{Alkida Balliu}
\author[2]{Gianlorenzo D'Angelo}
\author[1]{Pierre Fraigniaud\thanks{Additional supports from the INRIA project GANG. }}
\author[1,2]{Dennis Olivetti}

\affil[1]{ CNRS and University Paris Diderot, France.}
\affil[2]{Gran Sasso Science Institute, L'Aquila, Italy.}

\date{}

\maketitle

\begin{abstract}
In the framework of distributed network computing, it is known that, for every network predicate, each network configuration that satisfies this predicate can be proved using distributed certificates which can be verified locally. However, this requires to leak information about the identities of the nodes in the certificates, which might not be applicable in a context in which privacy is desirable. Unfortunately, it is known that if one insists on certificates independent of the node identities, then not all network predicates can be proved using distributed certificates that can be verified locally. In this paper, we prove that, for every network predicate, there is a distributed protocol satisfying the following two properties: (1)~for every  network configuration that is legal w.r.t.~the predicate, and for any attempt by an adversary to prove the illegality of that configuration using distributed certificates, there is a locally verifiable proof that the adversary is wrong, also using distributed certificates; (2)~for every  network configuration that is illegal w.r.t.~the predicate, there is a proof of that illegality, using distributed certificates, such that no matter the way an adversary assigns its own set of distributed certificates in an attempt to prove the legality of the configuration, the actual illegality of the configuration will be locally detected. In both cases, the certificates are independent of the identities of the nodes. These results are achieved by investigating the so-called \emph{local hierarchy} of complexity classes in which the certificates do not exploit the node identities. Indeed, we give a characterization of such a hierarchy, which is of its own interest
\end{abstract}

\vfill

\thispagestyle{empty}
\pagebreak
\setcounter{page}{1}
 
\section{Introduction}

\subsection{Context and objective}

In the framework of distributed network computing, local decision is the ability to check the legality of network configurations using a local distributed algorithm. This concern is of the utmost importance in the context of fault-tolerant distributed computing, where it is highly desirable that the nodes are able to collectively check the legality of their current configuration, which could have been altered by the corruption of variables due to failures. In this paper, we insist on locality, as we want the checking protocols to avoid involving long-distance communications across  the network, for they are generally costly and potentially unreliable. More specifically, we consider the standard LOCAL model of computation in networks~\cite{Pel00}. Nodes are assumed to be given  distinct identities, and each node executes the same algorithm, which proceeds in synchronous rounds where all nodes start at the same time. In each round, every node sends messages to its neighbors, receives messages from its neighbors, and performs some individual computation. The model does not limit the amount of data sent in the messages, neither it  limits the amount of computation that is performed by a node during a round. Indeed, the model places emphasis on the number of rounds before every node can output, as a measure of locality. (Note however that, up to some exceptions, our positive results involves messages of logarithmic  size, and polynomial-time computation). A \emph{local algorithm} is a distributed algorithm $\cA$ satisfying that there exists a constant $t\geq 0$ such that $\cA$ terminates in at most $t$ rounds in all networks, for all inputs. The parameter $t$ is called the \emph{radius} of $\cA$. In other words, in every network $G$, and for all inputs to the nodes of $G$, every node executing $\cA$ just needs to collect all information present in its $t$-ball around it  in order to output, where the \emph{$t$-ball} of $u$ is the ball $B_G(u,t)=\{v\in V(G):\dist(u,v)\leq t\}$. 

Following the guidelines of~\cite{FraKorPel13}, we define a \emph{configuration} as a pair $(G,x)$ where $G=(V,E)$ is a connected simple graph, and $x:V(G)\to\{0,1\}^*$ is a function assigning an input $x(u)$ to every node $u\in V$. A \emph{distributed language} $\cL$ is a Turing-decidable set of configurations. Note that the membership of a configuration to a distributed language is independent of the identity that may be assigned to the nodes in the LOCAL model (as one may want to study the same language under different computational models, including ones that assume anonymous nodes). The class $\LD$ is the set of all distributed languages that are locally decidable. That is, $\LD$ is the class of all distributed languages $\cL$ for which there exists a local algorithm $\cA$ satisfying that, for every configuration $(G,x)$, 
\[
(G,x)\in\cL \iff \cA\; \mbox{accepts} \; (G,x)
\] 
where one says that $\cA$ accepts if it accepts at \emph{all} nodes. More formally, given a graph $G$, let $\ID(G)$ denote the set of all possible identity assignments to the nodes of $G$ (with distinct non-negative integers). Then $\LD$  is the class of all distributed languages $\cL$ for which there exists a local algorithm $\cA$ satisfying the following: for every configuration $(G,x)$, 
 \[
\begin{array}{lcl}
(G,x)\in \cL & \Rightarrow & \forall \id\in\ID(G), \forall u \in V(G), \cA(G,x,\id,u) = \mbox{accept}\\
(G,x)\notin \cL & \Rightarrow & \forall \id\in\ID(G), \exists u \in V(G),  \cA(G,x,\id,u) = \mbox{reject}
\end{array}
\] 
where $\cA(G,x,\id,u)$ is the output of Algorithm~$\cA$ running on the instance $(G,x)$ with identity-assignment $\id$, at node~$u$. For instance, the language $\COL$, composed of all (connected) properly colored graphs, is in $\LD$. Similarly, the class $\LCL$ of ``locally checkable labelings'', defined in~\cite{NaoSto95}, satisfies $\LCL\subseteq \LD$. In fact, $\LCL$ is precisely $\LD$ restricted to configurations on graphs with constant maximum degree, and inputs of constant size. 

The class $\NLD$ is the non-deterministic version of $\LD$, i.e., the class of all distributed languages $\cL$ for which there exists a local algorithm $\cA$ \emph{verifying} $\cL$, i.e., satisfying that, for every configuration $(G,x)$, 
\[
(G,x)\in\cL \iff \exists c, \cA\; \mbox{accepts $(G,x)$ with certificate $c$}.
\] 
More formally, $\NLD$  is the class of all distributed languages $\cL$ for which there exists a local algorithm $\cA$ satisfying the following: for every configuration $(G,x)$, 
\[
\begin{array}{lcl}
(G,x)\in \cL & \Rightarrow & \exists c\in\cC(G), \forall \id\in\ID(G), \forall u \in V(G),  \cA(G,x,c,\id,u) = \mbox{accepts}\\
(G,x)\notin \cL & \Rightarrow & \forall c\in\cC(G), \forall \id\in\ID(G), \exists u \in V(G),  \cA(G,x,c,\id,u) = \mbox{rejects}
\end{array}
\] 
where $\cC(G)$ is the class of all functions $c:V(G)\to\{0,1\}^*$, assigning certificate $c(u)$ to each node~$u$. Note that the certificates $c$ may depend on the network and the input to the nodes, but should be set independently of the actual identity assignment to the nodes of the network. This guarantees that no information is revealed about far away nodes in the network by the certificates, hence preserving some form of privacy.  In the following, for the sake of simplifying the notations, we shall omit specifying the domain sets $\cC(G)$ and $\ID(G)$ unless they are not clear from the context. 

It follows from the above that $\NLD$ is a class of distributed languages that can be locally \emph{verified}, in the sense that, on legal instances, certificates can be assigned to nodes by a \emph{prover} so that a \emph{verifier} $\cA$ accepts, and, on illegal instances, the verifier $\cA$ rejects (i.e., at least one node rejects) systematically, and cannot be fooled by any fake certificate. For instance, the language $\TREE=\{(G,x): \mbox{$G$ is a tree}\}$ is in $\NLD$, by selecting a root $r$ of the given tree, and assigning to each node $u$ a counter $c(u)$ equal to its hop-distance to $r$. If the given (connected) graph contains a cycle, then no counters could be assigned  to fool an algorithm checking that, at each node $u$ with $c(u)\neq 0$, a unique neighbor $v$ satisfies $c(v)<c(u)$. In~\cite{FHK12}, $\NLD$ was proved to be exactly the class of distributed languages that are closed under lift. 

Finally, \cite{FraKorPel13} defined the randomized versions $\BPLD_{p,q}$ and $\BPNLD_{p,q}$, of the aforementioned classes $\LD$ and $\NLD$, respectively, by replacing the use of a deterministic algorithm with the use of a randomized algorithm characterized by its probability $p$ of acceptance for legal instances, and its probability $q$ of rejection for illegal instances. By defining $\BPNLD=\cup_{p^2+q\geq 1}\BPNLD_{p,q}$, the landscape of local decision was pictured as follows: 
\[
\LD \subset \NLD \subset \BPNLD = \all
\]
where all inclusions are strict, and $\all$ is the set of all distributed languages. That is, every distributed language can be locally verified with constant success probabilities $p$ and $q$, for some $p$ and $q$ satisfying $p^2+q\geq 1$. In other words, by combining non-determinism with randomization, one can decide any given distributed language. However, this holds only up to a certain guaranty of success, which is only guarantied to satisfy $p^2+q\geq 1$. 

\subsection{Our contributions}

Following up the approach recently applied to \emph{distributed graph automata} in~\cite{Reiter15}, we observe that the class $\LD$ and $\NLD$ are in fact the basic levels of a ``local hierarchy'' defined as follows. Let $\Sigma_0^{\scriptsize \local}=\Pi^{\scriptsize \local}_0=\LD$, and, for $k\geq 1$, let $\Sigma_k^{\scriptsize \local}$ be the class of all distributed languages $\cL$ for which there exists a local algorithm $\cA$ satisfying that, for every configuration~$(G,x)$, 
\[
(G,x)\in\cL \iff \exists c_1, \forall c_2, \dots, Q c_k, \cA\; \mbox{accepts $(G,x)$ with certificates $c_1,c_2,\dots,c_k$}
\] 
where the quantifiers alternate, and $Q$ is the universal quantifier if $k$ is even, and the existential one if $k$ is odd. The class $\Pi_k^{\scriptsize \local}$ is defined similarly, by starting with a universal quantifier, instead of an existential one. A local algorithm $\cA$ insuring membership to a class $\cC\in\{\Sigma_k^{\scriptsize \local}, k\geq 0\}\cup\{\Pi_k^{\scriptsize \local}, k\geq 0\}$ is called a $\cC$-algorithm. Hence, $\NLD=\Sigma_1^{\scriptsize \local}$, and, for instance, $\Pi_2^{\scriptsize \local}$  is the class of all distributed languages $\cL$ for which there exists a $\Pi_2^{\scriptsize \local}$-algorithm, that is, a local algorithm $\cA$ satisfying the following: for every configuration $(G,x)$, 
\begin{equation}
\begin{array}{lcl}
(G,x)\in \cL & \Rightarrow & \forall c_1,\exists c_2, \forall \id, \forall u \in V(G),  \cA(G,x,c_1,c_2,\id,u) = \mbox{accept;}\\
(G,x)\notin \cL & \Rightarrow & \exists c_1, \forall c_2, \forall \id, \exists u \in V(G),  \cA(G,x,c_1,c_2,\id,u) = \mbox{reject.}
\end{array}
\label{eq:pi2}
\end{equation}

\noindent 
Our main results are the following. 

\begin{theorem}\label{theo:main1}
$\LD \subset \Pi_1^{\scriptsize \local} \subset \NLD = \Sigma_2^{\scriptsize \local} \subset \Pi_2^{\scriptsize \local} = \all$, where all inclusions are strict. 
\end{theorem}

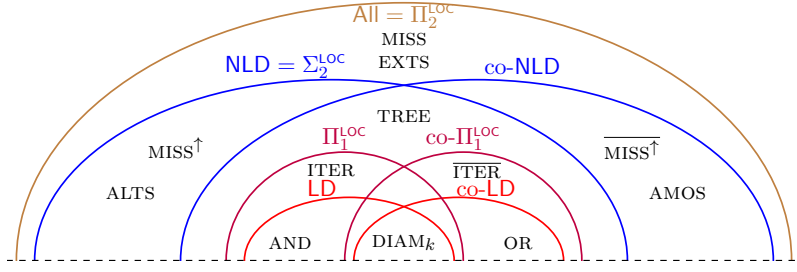
\begin{figure}[tb]
  \def\setld{(-1.2,0) ellipse (2.2cm and 2cm)}
  \def\setcold{(1.2,0) ellipse (2.2cm and 2cm)}
  \def\setnld{(-1.9,0) ellipse (5.5cm and 4cm)}
  \def\setconld{(1.9,0) ellipse (5.5cm and 4cm)}
  \def\setall{(0,0) ellipse (8.5cm and 5.7cm)}
  \centering
   \scalebox{0.8}{
   	\begin{tikzpicture}[scale=0.75]
   	\clip (-8.7,0) rectangle (8.7,6.5);
   	\draw[red,thick] (-1.2,0) ellipse (2.3cm and 1.4cm);
   	\draw (-1.8,1.6) node[text=red]{$\LD$};
   	\draw[red,thick] (1.2,0) ellipse (2.3cm and 1.4cm);
   	\draw (1.8,1.6) node[text=red] {$\coLD$};
   	\draw[purple,thick] (-1.3,0) ellipse (2.6cm and 2.4cm);
   	\draw (-1.3,2.7) node[text=purple]{$\Pi_1^{\scriptsize \local} $};
   	\draw[purple,thick] (1.3,0) ellipse (2.6cm and 2.4cm);
   	\draw (1.3,2.7) node[text=purple] {co-$\Pi_1^{\scriptsize \local} $};
   	\draw[blue,thick] (-1.6,0) ellipse (6.5cm and 4cm);
   	\draw(-2.6,4.3) node[text=blue] {$\NLD=\Sigma_2^{\scriptsize \local} $};
   	\draw[blue,thick] (1.6,0) ellipse (6.5cm and 4cm);
   	\draw(2.6,4.3) node[text=blue] {\coNLD};
   	\draw[brown,thick] (0,0) ellipse (8.5cm and 5.7cm);
   	\draw(0,5.4) node[text=brown] {$\all=\Pi_2^{\scriptsize \local} $};
   	\draw[dashed,thick] (-9,0) -- (9,0);
   	\draw (0,0.4) node{$\DIAM_k$};
   	\draw (-2.5,0.4) node{\AND};
  	\draw (2.5,0.4) node{\OR};
   	\draw (0.0,3.2) node{\TREE};
  	\draw (-6,1.5) node{\ALTS};
   	\draw (6,1.5) node{\AMOS};
   	\draw (-1.6,2) node{\ITER};
   	\draw (1.6,2) node{$\overline{\ITER}$};
   	\draw (0,4.4) node{\EXTS};
   	\draw (0,4.9) node{\MISS};
   	\draw (-5,2.5) node{$\MISSLIFT$};
   	\draw (5,2.5) node{$\overline{\MISSLIFT}$};
   	\end{tikzpicture}
   }
\caption{\sl Relations between the different decision classes of the local hierarchy (the definitions of the various languages can be found in the text).}
\label{fig:summary}
\end{figure}

That is, $\Pi_1^{\scriptsize \local}\supset \Pi_0^{\scriptsize \local}$, while $\Sigma_2^{\scriptsize \local}=\Sigma_1^{\scriptsize \local}$, and the whole local hierarchy collapses to the second level, at $ \Pi_2^{\scriptsize \local}$. This collapsing has a significant impact on our ability to certify the legality, or correctness of a network configuration w.r.t.~any (Turing decidable) boolean predicate on networks. Indeed, $\Pi_2^{\scriptsize \local} = \all$ says that, for every network predicate, there is a distributed protocol satisfying the following two properties: 
\begin{enumerate}
\item For every  network configuration that is legal w.r.t.~the predicate, and for any attempt by an adversary to prove the illegality of that configuration using distributed certificates, there is a locally verifiable proof that the adversary is wrong, also using distributed certificates, whose setting is \emph{independent of the node identities}. 
\item For every  network configuration that is illegal w.r.t.~the predicate, there is a proof of that illegality, using distributed certificates whose setting is \emph{independent of the node identities}, such that no matter the way an adversary assigns its own set of distributed certificates in an attempt to prove the legality of the configuration, the actual illegality of the configuration will be locally detected.
\end{enumerate}

We complete our description of the local hierarchy by a collection of separation and completeness results regarding the different classes and co-classes in the hierarchy. In particular, we revisit the completeness results in~\cite{FraKorPel13}, and show that the notion of reduction introduced in this latter paper is too strong, and may allow a language outside $\NLD$ to be reduced to a language in $\NLD$. We introduce a more restricted form of local reduction, called \emph{label-preserving}, which does not have this undesirable property, and we establish the following. 

\begin{theorem}\label{theo:main2}
 $\NLD$ and $\Pi_2^{\scriptsize \local}$ have complete distributed languages for local label-preserving reductions. \end{theorem}

\noindent Finally, Figure~\ref{fig:summary} summarizes all our separation results.

\subsection{Related Work}

Several form of ``local hierarchies'' have been investigated in the literature, with the objective of understanding the power of local computation, and/or for the purpose of designing verification mechanisms for fault-tolerant computing. In particular, \cite{Reiter15} has investigated the case of \emph{distributed graph automata}, where the nodes are finite automata, and the network is anonymous (which are weaker assumptions than those in our setting), but also assuming an arbitrary global interpretation of the individual decisions of the nodes (which is a stronger assumption than those in our setting). It is shown that all levels $\Sigma_k^{\scriptsize \aut}$, $k\geq 0$, of the resulting hierarchy are separated, and that the whole local hierarchy is exactly composed of the MSO (monadic second order) formulas on graphs. 

In the framework of distributed computing, where the computing entities are Turing machines, \emph{proof-labeling schemes} (PLS) \cite{KorKutPel10}, extended to \emph{locally checkable proofs} (LCP)~\cite{GS11}, give the ability to certify predicates using certificates that can take benefits of the node identities. That is, for the same network predicate, and the same legal network configuration, the distributed proof that this configuration is legal may be different if the node identities are different. In this context, the whole hierarchy  collapses at the first level, with $\Sigma_1^{\scriptsize \lcp}=\all$. However, this holds only if the certificates can be as large as $\Omega(n^2)$ bits. In~\cite{FFH16}, the class LogLCP~\cite{GS11}, which bounds the certificate to be of size $O(\log n)$ bits is extended to a hierarchy. In particular, it is shown that MST stands at the second level $\Pi_2^{\scriptsize \llcp}$ of that hierarchy, while there are languages outside the hierarchy. 

In \cite{FraKorPel13}, the authors introduced the model investigated in this paper. In particular, they defined and characterized the class $\NLD$, which is noting else than $\Sigma_1^{\scriptsize \local}$, that is, the class of languages that have a proof-labeling scheme in which the certificates are \emph{not} depending on the node identities. It is proved that, while $\NLD\neq\all$, randomization helps a lot, as the randomized version $\BPNLD$ of $\NLD$ satisfies $\BPNLD=\all$. It is also proved that, with the oracle \#nodes providing each node with  the number of nodes in the network, we get $\NLD^{\mbox{\footnotesize \#nodes}}=\all$. Interestingly, it was proved~\cite{FHK12} that restricting the verification algorithms for $\NLD$ to be \emph{identity-oblivious}, that is, enforcing that each node decides the same output for every identity-assignment to the nodes in the network, does not reduce the ability to verify languages. This is summarized by the equality $\NLDO=\NLD$ where the ``O'' in $\NLDO$ stands for identity-oblivious. Instead, it was recently proved that restricting the  algorithms to be identity-oblivious reduces the ability to decide languages locally, i.e., $\LDO\subsetneq \LD$ (see~\cite{FGKS13}). 

Finally, it is worth mentioning that the ability to decide a distributed language locally has impact on the ability to design \emph{construction} algorithms~\cite{Lin92} for that language (i.e., computing outputs $x$ such that the configuration $(G,x)$ is legal w.r.t. the specification of the task). For instance, it is known that if $\cL$ is locally decidable, then any randomized local construction algorithm for $\cL$ can be derandomized~\cite{NaoSto95}. This result has been recently extended~\cite{FF15} to the case of languages that are locally decidable by a randomized algorithm (i.e., extended from $\LD$ to $\BPLD$ according to the notations in~\cite{FraKorPel13}). More generally, the reader is invited to consult~\cite{FKKS09,KMW04,LOW08,LW08,Pel00,Suo13} for good introductions to local computing,  and/or  samples of significant results related to local computing.

\section{All languages are $\Pi_2^{\scriptsize \local}$ decidable}
\label{sec:allinPi2}

In this section, we show the last equality of Theorem~\ref{theo:main1}.

\begin{proposition}\label{prop:Pi2isall}
\em $\Pi_2^{\scriptsize \local} = \all$. 
\end{proposition}

\begin{proof}
Let $\cL$ be a distributed language.  We give an explicit $\Pi_2^{\scriptsize \local}$-algorithm for $\cL$, i.e., a local algorithm $\cA$ such that, for every configuration $(G,x)$, Eq.~\eqref{eq:pi2} is satisfied. For this purpose, we describe the distributed certificates~$c_1$ and~$c_2$. Intuitively, the certificate~$c_1$ aims at convincing each node that $(G,x)\not\in \cL$, while~$c_2$ aims at demonstrating the opposite. More precisely, at each node~$u$ in a configuration $(G,x)$, the certificate $c_1(u)$ is interpreted as a triple $(M(u),\data(u),\ind(u))$ where~$M(u)$ is an $m \times m$ boolean matrix, $\data(u)$ is a linear array with $m$ entries, and $\ind(u)\in\{1,\dots,m\}$. Informally, $c_1(u)$ aims at proving to node~$u$ that it is node labeled $\ind(u)$ in the $m$-node graph with adjacency matrix~$M(u)$, and that the whole input data is~$\data(u)$.  We denote by~$n$ the number of nodes of the actual graph~$G$. 

For a legal configuration $(G,x)\in \cL$, given $c_1$, the certificate~$c_2$  is then defined as follows. It is based on the identification of a few specific nodes, that we call \emph{witnesses}. Intuitively, a witness is a node enabling to demonstrate that the structure of the configuration $(G,x)$ does not fit with the given certificate $c_1$. Let $\dist(u,v)$ denote the distance between any two nodes $u$ and $v$ in the actual network~$G$, that is, $\dist(u,v)$ equals the number of edges of a shortest path between $u$ and $v$ in $G$. A certificate $c_2(u)$ is of the form $(f(u),\sigma(u))$ where $f(u)\in\{0,\dots,4\}$ is a flag, and $\sigma(u)\in\{0,1\}^*$ depends on the value of the flag. 

\medskip
\noindent  -- Case 0: There are two adjacent nodes $v\neq v'$ such that $(M(v),\data(v))\neq (M(v'),\data(v'))$, or there are nodes $v$ in which $c_1(v)$ cannot be read as a triple $(M(v),\data(v),\ind(v))$. Then we set one of these nodes as witness~$w$, and we set $c_2(u) = (0,\dist(u,w))$ at every node~$u$. 

\medskip

\noindent Otherwise, i.e., if the pair $(M(u),\data(u))$ is identical to some pair $(M,\data)$ at every node~$u$:

%
\medskip
\noindent -- Case 1: $(G,x)$ is isomorphic to $(M,\data)$, preserving the inputs, denoted by $(G,x) \sim (M,\data)$, and $\ind()$ is injective. Then we set $c_2(u) = (1)$ at every node~$u$. 

\noindent -- Case 2: $n>m$, i.e., $|V(G)|$ is larger than the dimension $m$ of $M$, or $\ind()$ is not injective. Then we set the certificate $c_2(u) = (2,i,d(u,w),d(u,w'))$ where $i\in\{1,\dots,m\}$, and $w\neq w'$ are two distinct nodes such that $\ind(w)=\ind(w')=i$. These two nodes $w$ and $w'$ are both witnesses. 

\noindent -- Case 3: $n<m$ and $\ind()$ is injective. Then we set $c_2(u) = (3,i)$ where $i\in\{1,\dots,m\}$ is such that $\ind(v)\neq i$ for every node~$v$. 

\noindent -- Case 4: $n=m$ and $\ind()$ is injective, but $(G,x)$ is not isomorphic to $(M,\data)$. Then we set as witness  a node~$w$ whose neighborhood in $(G,x)$ does not fit with what it should be according to  $(M,\data)$, and we set $c_2(u) = (4,d(u,w))$ for every node~$u$. 
\medskip

The local verification algorithm $\cA$ then proceeds as follows. First, every node $u$  checks whether its flag $f(u)$ in $c_2(u)$ is identical to all the ones of its neighbors, and between~0 and~4. If not, then $u$ rejects. Otherwise, $u$ carries on executing the verification procedure. Its behavior depends on the value of its flag. 

\medskip
\noindent -- If $f(u)=0$, then $u$ checks that at least one of its neighbors has a distance to the witness that is smaller than its own distance. A node with distance $0$ to the witness checks that there is indeed an inconsistency with its $c_1$ certificate (i.e., its $c_1$ certificate cannot be read as a pair matrix-data, or its $c_1$ certificate is distinct from the one of its neighbors). Every node accepts or rejects accordingly. 

\noindent -- If $f(u)= 1$, then $u$ accepts or rejects according to whether $(M(u),\data(u))\in\cL$ (recall that, by definition, we consider only distributed languages $\cL$ that are Turing-decidable).

\noindent -- If $f(u)= 2$, then $u$ checks that it has the same index $i$ in its certificate $c_2$ as all its neighbors. If that is not the case, then it rejects. Otherwise, it checks each of the two distances in its certificate $c_2$ separately, each one as in the case where $f(u)=0$. A node with one of the two distances equal to $0$ also checks that its $c_1$ index is equal to the index~$i$ in $c_2$. If that is not the case, or if its two distances are equal to~$0$, then it rejects. If all the test are passed, then $u$ accepts. 

\noindent -- If $f(u)=3$, then $u$ accepts if and only if it has the same index $i$ in its $c_2$ certificate as all its neighbors, and $\ind(u)\neq i$.

\noindent -- If $f(u)=4$, then $u$ checks the distances as in the case where $f(u)=0$.  A node with distance~0 also checks that its neighborhood in the actual configuration $(G,x)$ is not what it should be according to $(M,\data)$. It accepts or rejects accordingly. 
\medskip

To prove the correctness of this Algorithm~$\cA$, let us first consider a legal configuration $(G,x)\in \cL$. We show that the way $c_2$ is defined guarantees that all nodes accept, because $c_2$ correctly pinpoints inconsistencies in $c_1$, witnessing any attempt of $c_1$ to certify that the actual configuration is illegal. Indeed, in Case~0, by the setting of $c_2$, all nodes but the witness accept. Also, the witness itself accepts because it does witness the inconsistency of the $c_1$ certificate. In Case~1, then all nodes accept because $(G,x) \sim (M,\data)$ and $(G,x)\in \cL$. In Case~2, by the setting of $c_2$, all nodes but the witnesses accept, and  the witnesses  accept too because each one  checks that it is the vertex with index~$i$ in $M$. In Case~3, all nodes accept by construction of the certificate~$c_2$. Finally, in Case~4, by the setting of $c_2$, all nodes but the witness accept. Also, the witness itself accepts because, as in Case~0, it does witness the inconsistency of the $c_1$ certificate. So, in all cases, all nodes accept, as desired. 

We are now left with the case of illegal configurations. Let  $(G,x)\notin \cL$ be such an illegal configuration. We set $c_1(u)=(M,\data,\ind(u))$ where $(M,\data)\sim (G,x)$ and $\ind(u)$ is the index of node $u$ in the adjacency matrix $M$ and the array $\data$. We show that, for any certificate $c_2$, at least one node rejects. Indeed, for all nodes to accept, they need to have the same flag in $c_2$. This flag cannot be~1 because, if $f(u)=1$ then $u$ checks the legality of $(M,\data)$. In all other cases, the distance checking should be passed at all nodes for them to accept. Thus, the flag is distinct from~0 and~4 because every radius-1 ball in $(G,x)$ fits with its description in $(M,\data)$. Also, the flag is distinct from~2 because there are no two distinct nodes with the same index~$i$ in the $c_1$ certificate. Finally, also the flag is distinct from~3, because, by the setting of $c_1$, every index in $\{1,\dots,n\}$ appears at some node, and this node would reject. Hence, all cases lead to contradiction, that is, not all nodes can accept, as desired. 
\end{proof}

To conclude the section, let us define a simple problem in $\Pi_2^{\scriptsize \local}\setminus \NLD$. Let $\EXTS$, which stands for ``exactly two selected'' be the following language. We set $(G,x)\in\EXTS$ if $x(u)\in\{\bot,\top\}$ for every $u\in V(G)$, and $|\{u\in V(G): x(u)=\top\}|= 2$. Proving that $\EXTS\notin \NLD$ is easy using the following characterization of $\NLD$. 
Let $t\geq 1$. A configuration $(G',x')$ is a $t$-lift of a configuration $(G,x)$ iff there exists a mapping $\phi:V(G')\to V(G)$ such that, for every $u\in V(G')$, $B_{G}(\phi(u),t)$ is isomorphic to $B_{G'}(u,t)$, preserving inputs. A distributed language $\cL$ is closed under lift if there exists $t\geq 1$ such that, for every $(G,x)$, we have $(G,x)\in \cL$ implies $(G',x')\in \cL$ for every $(G',x')$ that is a $t$-lift of $(G,x)$. 

\begin{lemma}[\cite{FHK12}]\label{lem:lift}
$\NLD$ is the class of distributed languages closed under lift. 
\end{lemma}

Since $\EXTS$ is not closed under lift, it results from Lemma~\ref{lem:lift} that $\EXTS\notin \NLD$. 

\section{On the impact of the last universal quantifier}

In this section, we prove the part of Theorem~\ref{theo:main1} related to the two classes $\Pi_1^{\scriptsize \local}$ and $\Sigma_2^{\scriptsize \local}$. These two classes have in common that the universal quantifier is positioned last. It results that these two classes seem to be limited, as witnessed by the following two propositions.

\begin{proposition}\em\label{prop:sigma2-equal-NLD}
 $\Sigma_2^{\scriptsize \local} = \NLD$. 
\end{proposition}

\begin{proof}
By Lemma~\ref{lem:lift}, it is sufficient to prove that, for any $\cL\in \Sigma_2^{\scriptsize \local}$, $\cL$ is closed under lift. If $\cL\in \Sigma_2^{\scriptsize \local}$ then let $\cA$ be a local algorithm establishing the membership of $\cL$ in $\Sigma_2^{\scriptsize \local}$. $\cA$ is satisfying the following. For every  $(G,x)\in \cL$, 
\[
\exists c_1,\forall c_2, \forall \id, \forall u \in V(G),  \cA(G,x,c_1,c_2,\id,u) = \mbox{accept.}
\]
Let $t$ be the radius of $\cA$, and assume, for the purpose of contradiction, that $\cL$ is not closed under lift. There exists $(G,x)\in\cL$, and a $t$-lift $(G',x')$ of $(G,x)$ with $(G',x')\notin \cL$. Let $\phi:V(G')\to V(G)$ be this $t$-lift. Note that the $t$-balls in $(G,x)$ are identical to the $t$-balls in $(G',x')$ by definition of a $t$-lift. Let $c_1$ be the distributed certificate that makes $\cA$ accept $(G,x)$ at all nodes, for all certificates $c_2$. Let $c'_1$ be the distributed certificate for $(G',x')$ defined by $c'_1(u)=c_1(\phi(u))$. Since, with certificate $c_1$, $\cA$ accepts at all nodes of $G$, for every certificate $c_2$, and for every identity assignment, it follows that, with certificate $c'_1$, $\cA$ accepts at all nodes of $G'$, for every certificate $c'_2$, and for every identity assignment, contradicting the correctness of $\cA$. Therefore,  $\cL$ is closed under lift. Thus, $\Sigma_2^{\scriptsize \local} \subseteq \NLD$. Since, by definition, $\NLD=\Sigma_1^{\scriptsize \local} \subseteq \Sigma_2^{\scriptsize \local}$, the result follows. 
\end{proof}

To show that $\Pi_1^{\scriptsize \local} \neq \NLD$, we consider the language $\ALTS$, which stands for ``at least two selected''. (Note that $\ALTS$ is the complement of the language $\AMOS$ introduced in~\cite{FraKorPel13}, where $\AMOS$ stands for ``at most one selected''). We set $(G,x)\in\ALTS$ if $x(u)\in\{\bot,\top\}$ for every node $u\in V(G)$, and $|\{u\in V(G): x(u)=\top\}|\geq 2$. Note that $\ALTS\in \NLD\setminus \Pi_1^{\scriptsize \local}$.

\begin{proposition}\em\label{prop:Pi1subsetNLD}
$\Pi_1^{\scriptsize \local} \subset \NLD$ (the inclusion is strict). 
\end{proposition}

\begin{proof}
By Lemma~\ref{lem:lift}, to establish  $\Pi_1^{\scriptsize \local} \subseteq \NLD$, it is sufficient to prove that, for any $\cL\in \Pi_1^{\scriptsize \local}$, $\cL$ is closed under lift. The arguments are exactly similar to the ones used in the proof of Proposition~\ref{prop:sigma2-equal-NLD} without even the need to lift a first set of certificates. To show that $\Pi_1^{\scriptsize \local} \neq \NLD$, we consider the language $\ALTS$. We have $\ALTS\in \NLD$ because $\ALTS$ is closed under lift. However, $\ALTS\notin \Pi_1^{\scriptsize \local}$. Indeed, assume that there exists a local algorithm $\cA$ for $\ALTS\in\Pi_1^{\scriptsize \local}$. Then, consider the cycle $C_n$, and three distinct nodes $u,v,w$ in $C_n$ equally spread at distance $n/3$. Let us fix a set of certificates, say $c(u)=\emptyset$ for every node $u$. $\cA$ cannot be correct for all eight configurations resulting from the two possible input assignments $\bot$ or $\top$ to the three nodes $u,v,w$. 
\end{proof}

While $\Pi_1^{\scriptsize \local}$ is in $\NLD$, the universal quantifier does add some power compared to $\LD$. We show that $\LD \neq \Pi_1^{\scriptsize \local}$ by exhibiting a language in $\Pi_1^{\scriptsize \local}\setminus \LD$. Note that the existence of this language is not straightforward as it must involve Turing-computability issues. Indeed, if one does not insist on the fact that the local algorithm must be a Turing-computable function, then the two classes $\LD$ and $\Pi_1^{\scriptsize \local}$ would be identical. For instance, given a $t$-round algorithm $\cA$ deciding  a language $\cL$ in $\Pi_1^{\scriptsize \local}$, one could define the following function $f$ for deciding the same  language in $\LD$. Given a $t$-ball $B$ centered at $u$, node $u$ accepts with $f$ if and only if there are no certificate assignments to the nodes of $B$ that could lead $\cA$ to reject at $u$. The function $f$ is however not an algorithm. We show that, in fact,  $\Pi_1^{\scriptsize \local}\setminus \LD\neq \emptyset$.

\begin{proposition}\em
 $\LD \subset \Pi_1^{\scriptsize \local}$  where the inclusion is strict.
\end{proposition}

\begin{proof}
We describe the distributed language $\ITER$, which stands for ``iteration''. Let $M$ be a Turing machine, and let us enumerate lexicographically all the states of the system tape-machine where $M$ starts its execution on the blank tape, with the head at the beginning of the tape. We define the function $f_M:\mathbb{N}\to\mathbb{N}$ by $f_M(0)=0$, $f_M(1)=1$, and, for $i>1$, $f_M(i)$ equal to the index of the system state after one step of $M$ from system state $i$. We define $\ITER$ as the collection of configurations $(G,x)$ representing two sequences of iterations of a function $f_M$ on different inputs $a$ and $b$ (see Figure~\ref{fig:iter}). 

\begin{figure}[htb]
\centerline{\includegraphics[scale=1.3]{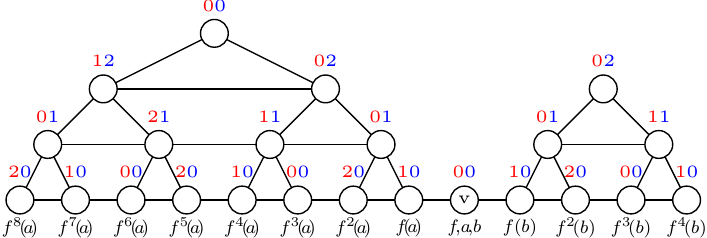}}
\caption{\sl An illustration of the distributed language $\ITER$}
\label{fig:iter}
\end{figure}

More precisely, let $M$ be a Turing machine, and let $a$ and $b$ be two non-negative integers. We define the following family of configurations --- cf. Figure~\ref{fig:iter}. A configuration in $\ITER$ mainly consists in a path $P$ with a special node $v$, called the \emph{pivot}, identified in this path. So $P=LvR$ where $L$ and $R$ are subpaths, respectively called left path and right path. All nodes of the path are given the machine $M$ as input, and the pivot $v$ is also given $a$ and $b$ as inputs. The node of the left path (resp., right path) at distance $i$ from $v$ is given a value $f_{i,L}$ (resp., $f_{i,R}$) as input. To be in the language, it is required that, for every $i$, 
$
f_{i,L}=f_M^{(i)}(a) \; \mbox{and} \; f_{i,R}=f_M^{(i)}(b), 
$
where $g^{(i)}$ denotes the $i$th iterated of a function $g$. Let $u_\ell$ and $u_r$ be the two nodes at the extremity of the left path and of the right path, respectively. The configuration is in the language if and only if the $f$-values at both extremities of the path $P$ are~0 or~1, and at least one of them is equal to~0. That is, the configuration is in the language if and only if:
\begin{equation}\label{eq:itercond}
(f_{|L|,L}\in\{0,1\} \; \mbox{and} \; f_{|R|,R}\in\{0,1\}) \;\mbox{and} \; (f_{|L|,L}=0 \; \mbox{or} \; f_{|R|,R}=0).
\end{equation} 
In fact, for technical reasons, it is also required that both $|L|$ and $|R|$ are powers of~2. Indeed, on top of $L$ and $R$ are two complete binary trees $T_L$ and $T_R$, respectively, with horizontal paths connecting nodes of the same depth in each tree (see Figure~\ref{fig:iter}). The nodes of $L$ and $R$ are the leaves of these two trees. Finally, every node $u$ of the  graph receives as input a pair of labels $(\ell_1,\ell_2)\in\{0,1,2\}^2$. The label $\ell_1$ is the distance modulo~3 from $u$ to the right-most node (resp., left-most node) of the path if $u$ is an internal  node of $T_L$ (resp., $T_R$), and, for nodes in the path $P$, $\ell_1$ is simply the distance modulo~3 from the pivot $v$. The label $\ell_2$ is the height of the node in its tree modulo~3. (The pivot, which belongs to none of the trees, has height~0). 

A configuration $(G,x)\in \ITER$ if and only if $(G,x)$ satisfies all the above conditions with respect to the given machine~$M$. 
Let us consider a weaker version of $\ITER$, denoted by $\ITER^-$ where the condition of Eq.~\eqref{eq:itercond} is replaced by just: 
$
f_{|L|,L}\in\{0,1\} \; \mbox{and} \; f_{|R|,R}\in\{0,1\}. 
$
Thanks to the labeling $(\ell_1,\ell_2)$ at each node, which ``rigidifies'' the structure, we have $\ITER^-\in \LD$ using the same arguments as the ones in~\cite{FGKS13}.  
Moreover, $\ITER\in \Pi_1^{\scriptsize \local}$. To see why, we describe a local algorithm $\cA$ using certificates. The algorithm first checks whether $(G,x)\in \ITER^-$. All nodes, but the pivot $v$, decide according to this checking. If the pivot rejected $(G,x)\in \ITER^-$, then it rejects in $\cA$ as well. Otherwise, it carries on its decision process by interpreting its certificate as a non-negative integer $k$, and accepts in $\cA$ unless $f_M^{(k)}(a)=1$ and $f_M^{(k)}(b)=1$. To show the correctness of $\cA$, let $(G,x)\in\ITER$. We have $f_{|L|,L}=0$ or $f_{|R|,R}=0$, i.e., $f_M^{(|L|)}(a)=0$ or $f_M^{(|R|)}(b)=0$. W.l.o.g., assume $f_M^{(|L|)}(a)=0$. If $k\geq |L|$ then $f_M^{(k)}(a)=0$ since $f_M(0)=0$, and thus $v$ accepts. If $k < |L|$ then $f_M^{(k)}(a)\neq 1$ since $f_M(1)=1$, and thus $v$ accepts. Therefore, all certificates lead to acceptance. Let us now consider $(G,x)\notin\ITER$. If $(G,x)\notin\ITER^-$ then at least one node rejects, independently of the certificate. So, we assume that  $(G,x)\in\ITER^-\setminus\ITER$. Thus, $f_M^{(|L|)}(a)=1$ and $f_M^{(|R|)}(b)=1$. The certificate is set to $k=\max\{|L|,|R|\}$. Let us assume, w.l.o.g., that $k=|L|\geq |R|$.  By this setting, we have $f_M^{(k)}(a)=1$. Moreover, since $k\geq |R|$, and since $f_M(1)=1$, we get that $f_M^{(k)}(b)=1$. Therefore, $\cA$ rejects, as desired. Thus, $\ITER\in\Pi_1^{\scriptsize \local}$. 

It remains to show that $\ITER\notin\LD$. Let us assume, for the purpose of contradiction, that there exists a $t$-round algorithm  $\cA$ deciding $\ITER$. Since $\ITER^-\in\LD$, this algorithm is able to distinguish an instance with $f_M^{(|L|)}(a)=1$ and $f_M^{(|R|)}(b)=1$ from instances in which $f_M^{(|L|)}(a)\neq 1$ or $f_M^{(|R|)}(b)\neq 1$. Observe that a node at distance greater than $t$ from the pivot can gather  information related to only one of the two inputs $a$ and $b$. Therefore, the distinction between the case $f_M^{(|L|)}(a)=1$ and $f_M^{(|R|)}(b)=1$ and the case $f_M^{(|L|)}(a)\neq 1$ or $f_M^{(|R|)}(b)\neq 1$ can only be made by a node at distance at most $t$ from the pivot. Therefore, by simulating $\cA$ at all nodes in the ball of radius $t$ around $v$, with identities between~1 and the size of the ball of radius $2t$ around the pivot, a sequential algorithm can determine, given a Turing machine $M$, and given $a$ and $b$, whether  there exist $\ell$ and $r$ such that $f_M^{(\ell)}(a)=f_M^{(r)}(b)=1$ or not, which is actually Turing undecidable. This contradiction implies that, indeed,  $\ITER\notin\LD$. 
\end{proof}

\section{Complement classes}

Given a class $\cC$ of distributed languages, the class co-$\cC$ is composed of all distributed languages $\cL$ such that $\bar{\cL}\in\cC$, where $\bar{\cL}=\{(G,x)\notin \cL\}$. For instance, 
%
co-$\Pi_1^{\scriptsize \local}$ is the class of languages $\cL$ for which there exists a local algorithm $\cA$ such that, for every configuration $(G,x)$, 
\[
\begin{array}{lcl}
(G,x)\in \cL & \Rightarrow & \exists c, \forall \id, \exists u \in V(G),  \cA(G,x,c,\id,u) = \mbox{accepts;}\\
(G,x)\notin \cL & \Rightarrow & \forall c, \forall \id, \forall u \in V(G),  \cA(G,x,c,\id,u) = \mbox{rejects.}\\
\end{array}
\]
Note in particular, that the rejection must now be unanimous, while the acceptance requires only one node to accept. Let us define the following two languages: 
each input to every node belongs to $\{\mbox{true},\mbox{false}\}$, and a configuration is in $\AND$ (resp., in $\OR$) if and only if the logical conjunction (resp., disjunction) of the inputs is true. 
%
We have $\overline{\AND}=\OR$. These two languages enable to separate $\LD$ from its co-class. Indeed, trivially, $\OR \notin \LD$, and $\AND\in\LD$ (without communication).  The class $\LD\cap\coLD$ is almost trivial. It contains simple distributed languages such as $\DIAM_k$, the class of graphs with diameter at most $k$, for any fixed $k$. We have the following separation. 

\begin{proposition}\em
$\OR \in \coLD \setminus \Pi_1^{\scriptsize \local}$, and $\AND \in \LD\setminus \mbox{\rm co-}\Pi_1^{\scriptsize \local}$. 
\end{proposition}

Similarly, the languages $\ALTS$ and $\AMOS$ introduced in the proof of Proposition~\ref{prop:Pi1subsetNLD} enable to separate $\NLD$ from its co-class. Indeed, $\ALTS=\overline{\AMOS}$, $\ALTS$ is closed under lift, and $\AMOS$ is not closed under lift. Moreover, consider the language $\EXTS$ defined at the end of Section~\ref{sec:allinPi2}. Both $\EXTS$ and $\overline{\EXTS}$ are not closed under lift. So, overall, by Lemma~\ref{lem:lift}, we get: 

\begin{proposition}\em
 $\ALTS \in \NLD \setminus \coNLD$, $\AMOS \in \coNLD \setminus \NLD$, and $\EXTS\notin \NLD\cup\coNLD$. 
\end{proposition}

More interesting is the position of the $\Pi_1^{\scriptsize \local}$ w.r.t. $\NLD$ and $\coNLD$: 

\begin{proposition}\label{prop:unionPiintersectNLD}\em
$\Pi_1^{\scriptsize \local} \cup \mbox{\rm co-}\Pi_1^{\scriptsize \local} \subset \NLD \cap \coNLD$, where the inclusion is strict.
\end{proposition}

\begin{proof}
From Proposition~\ref{prop:Pi1subsetNLD}, we know that $\Pi_1^{\scriptsize \local} \subset \NLD$. We prove that $\mbox{\rm co-}\Pi_1^{\scriptsize \local} \subset \NLD$. Let $\cL\in \mbox{\rm co-}\Pi_1^{\scriptsize \local}$, and let $\cA$ be a $t$-round algorithm deciding $\cL$. Let $(G,x)\in \cL$, and let $c$ be a certificate such that $\cA$ accepts at all nodes. Let $(G',x')$ be a $t$-lift of $(G,x)$, and lift $c$ into $c'$ accordingly. Then $\cA$ also accepts $(G',x')$, which implies that $\cL$ is closed under $t$-lift, and thus, by Lemma~\ref{lem:lift}, $\cL\in \NLD$. Therefore
$
\Pi_1^{\scriptsize \local} \cup \mbox{\rm co-}\Pi_1^{\scriptsize \local} \subseteq \NLD \cap \coNLD.
$
To prove that the inclusion is strict, we consider the language $\TREE=\{(G,x):G\;\mbox{is a tree}\}$. We have $\TREE\in\NLD$ since a tree cannot be lifted. We also have $\TREE\in\coNLD$ since a tree cannot result from a lift. (By Lemma~\ref{lem:lift}, $\coNLD$ is the class of languages $\cL$ closed down under lift, i.e., if $(G,x)\in\cL$ is the lift of a configuration $(G',x')$, then we have $(G',x')\in \cL$). To see why $\TREE\notin\Pi_1^{\scriptsize \local}$, consider a path and a cycle. If $\TREE$ could be decided in $\Pi_1^{\scriptsize \local}$, then the center nodes of the path must accept for all certificates and for any identity-assignment. Hence, all degree-2 nodes that see only degree-2 nodes in their neighborhoods accept, for all certificates. As a consequence, the cycle will be incorrectly accepted for all certificates. Somewhat similarly,  if $\TREE$ could be decided in co-$\Pi_1^{\scriptsize \local}$, say in $t$-rounds, then it would mean that, in a path, the node(s) that accept(s) (with the appropriate certificate) can only be at distance at most $t$ from an extremity of the path. Indeed, otherwise, one could close the path and create a cycle that will still be accepted. So, by gluing two paths $P$ and $P'$ of length at least $2t$ to two antipodal nodes of a cycle $C$, and by giving to the nodes of $P$ and $P'$ the certificates that lead each of them to be accepted, this graph would be incorrectly accepted.  
\end{proof}

\section{Complete problems}
 
In this section, we prove Theorem~\ref{theo:main2}.  Let $G$ be a connected graph, and $U$ be a set (typically, $U=\{0,1\}^*$). Let $e: V(G)\to U$, and let $\cS:V(G) \to 2^{2^U}$. That is, $e$ assigns  an element $e(u)\in U$ to every node $u\in V(G)$, and $\cS$ assigns a collection of sets $\cS(u)=\{S_1(u),\dots,S_{k_u}(u)\}$ to every node $u\in V(G)$, with $k_u\geq 1$ and $S_i:V(G)\to 2^U$ for every $i\geq 1$. We say that $\cS$ covers $e$ if and only if there exists $u\in V(G)$, and there exists $i\in\{1,\dots,k_u\}$, such that $S_i(u)=\{e(u), u\in V(G)\}$. In~\cite{FraKorPel13}, the authors defined the language 
\[
\COVER=\{(G,x): \forall u\in V(G), x(u)=(\cS(u),e(u)) \;\mbox{such that $\cS$ covers $e$}\}
\]
and proved that $\COVER$ is the ``most difficult problem'', in the sense that every distributed language can be locally reduced to $\COVER$. However $\COVER$ is closed under lift as lifting does not create new elements, while lifting preserves the sets. Therefore, by Lemma~\ref{lem:lift},  $\COVER\in\NLD$. In fact, one can show that there exists a local verification algorithm for $\COVER$ using certificates of size quasi linear in $n$ whenever the ground set $U$ is of polynomial size. 

\begin{proposition}\em\label{prop:smallcertif}
Let $U$ be the ground set of  $\COVER$. Then  $\COVER$ has a local decision algorithm for $\NLD$, using certificates of size  $O(n(\log n+\log |U|))$ bits.
\end{proposition}

\begin{proof}
Given $(G,x)\in\COVER$, where $G$ is an $n$-node graph, the prover assigns the following certificates to the nodes. For any $u\in V(G)$, we have $c(u)=(d_0,(d_1,e_1),(d_2,e_2),\dots,(d_n,e_n))$, where, for every $i\in\{0,\dots,n\}$, $d_i$ is a non-negative integer, and, for every $i\in\{1,\dots,n\}$, $e_i\in U$. This certificate is on $O(n(\log n+\log |U|))$ bits. The $d_i$'s measure distances: $d_0$ is the distance from $u$ to the node $v$ which has a set $S_i(v)$ covering $e$, and, for every $i\in\{1,\dots,n\}$, $d_i$ is the distance from $u$ to a node $u'$ with $e(u')=e_i$.  

The verifier acts as follows, in just one communication round. Every node $u$ checks that it has the same number of distance entries in its certificate as all its neighbors, and that the $i$th elements coincide between neighbors, for every $i=1,\dots,n$. Next, it checks that one and only one of its distances $d_i$ with $i\in\{1,\dots,n\}$ is null, and that $e(u)=e_i$.  Next, if $d_0=0$, it checks that it has a set $S_j(u) = \{e_i, i=1,\dots,n\}$. Finally, it checks that the distances are consistent, that is, for every $i$ such that $d_i\neq 0$, it checks that it has at least one neighbor whose $i$th distance is  smaller than $d_i$. If all tests are passed, then $u$ accepts, otherwise it rejects. 

By construction, if $(G,x)\in\COVER$ then all nodes accept. Conversely, let us assume that all nodes accept. Since the distances are decreasing, for each element $e_i$ there must exist at least one node $u$ such that $x(u)=e_i$. Conversely, every node has its element appearing in the certificate (because it must have one distance equal to~0). Finally, since the distance $d_0$ is decreasing, there must exist at least one node $u$ that has a set $S_j(u) = \{e_i, i=1,\dots,n\}$. This implies that $(G,x)\in\COVER$. 
\end{proof}

This latter result is in contradiction with the claim in~\cite{FraKorPel13} regarding the hardness of $\COVER$. The reason for this contradiction is that the local reduction used in~\cite{FraKorPel13} for reducing any language to $\COVER$ is too strong. Indeed,  it transforms a configuration $(G,x)$ into a configuration $(G,x')$ where the certificates used for proving $x'$ may depend on the identities of the nodes in $G$. This is in contradiction with the definitions of the  classes $\Sigma_k^{\scriptsize \local}$ and $\Pi_k^{\scriptsize \local}$, $k\geq 0$, for which the certificates must be independent of the identity assignment. In this section, we show that completeness results can be obtained using a more constrained notion of reduction which preserves the membership to the classes. 

Recall from \cite{FraKorPel13} that a local reduction of $\cL$ to $\cL'$ is a local algorithm $\cR$ which maps any configuration $(G,x)$ to a configuration $(G,y)$, where $y=R(G,x,\id)$ may depend of the identity assignment $\id$, such that: $(G,x)\in \cL$ if and only if, for every identity assignment $\id$ to the nodes of $G$, $(G,y) \in \cL'$ where $y=\cR(G,x,\id)$. Ideally, we would like $\cR$ the be \emph{identity-oblivious}, that is, such that the output of each node does not depend on the identity assignment. We use a concept somewhat intermediate between identity-oblivious reduction and the unconstraint reduction in~\cite{FraKorPel13}. 

\begin{definition}\em \label{def:reductino}
Let $\cC$ be a class of distributed languages, and let $\cL$ and $\cL'$ be two distributed languages.  Let $\cA$ be a   $\cC$-algorithm  deciding $\cL'$, and let $\cR$ be a local reduction of $\cL$ to $\cL'$. We say that $(\cR,\cA)$ is \emph{label-preserving} for $(\cL,\cL')$ if and only if, for any configuration $(G,x)$, the certificates used by the prover in $\cA$ for $(G,y)$ where $y=\cR(G,x,\id)$ are the same for all  identity assignments $\id$ to $G$. 
\end{definition}

The following result shows that the notion of reduction in Definition~\ref{def:reductino} preserves the classes of distributed languages.

\begin{lemma}\label{lem:preserveclasses}
Let $\cC$ be a class of distributed languages. Let $\cL$ and $\cL'$ be two distributed languages with $\cL'\in \cC$, and let $(\cR,\cA)$ be a label-preserving local reduction for $(\cL,\cL')$. Then $\cL\in \cC$.
\end{lemma}

\begin{proof}
We describe a  local algorithm $\cB$ for deciding $\cL$ in $\cC$. In essence, $\cB=\cA\circ R$. More precisely, let $(G,x)$ be a configuration, with an arbitrary identity assignment $\id$, and let $y=R(G,x,\id)$. Let $c$ be a certificate assigned by the prover in $\cA$ for configuration $(G,y)$. (Note that this certificate may depend on some previously set certificates, as in, e.g., $\Pi_1^{\scriptsize \local}$). The certificate assigned by the prover in $\cB$ for configuration $(G,x)$ is $c$. The algorithm $\cB$ then proceeds as follows. Given $(G,x)$, it computes $(G,y)$ using $R$, and then applies $\cA$ on $(G,y)$ using the certificates constructed by the prover in $\cB$. Algorithm $\cB$ then outputs the decision taken by $\cA$. Since $R$ preserves the membership to the languages, and since the certificates assigned by the prover in $\cA$ for configurations resulting from the application of $R$ are independent of the identity assignment, the certificates chosen under the identity assignment $\id$ are also appropriate for any other identity assignment $\id'$. This guarantees the correctness of $\cB$, and thus $\cL\in\cC$.
\end{proof}

In the following problem, every node $u$ of a configuration $(G,x)$ is given a family  $\cF(u)$ of configurations, each described by an adjacency matrix representing a graph, and a 1-dimensional array representing the inputs to the nodes of that graph. In addition, every node $u$ has an input string $x'(u)\in\{0,1\}^*$. Hence, $(G,x')$ is also a configuration. The actual configuration $(G,x)$ is legal if $(G,x')$ is missing in all families $\cF(u)$ for every $u\in V(G)$, i.e., $(G,x')\notin \cF$ where $\cF=\cup_{u\in V(G)}\cF(u)$. In short,  we consider the language
\[
\MISS=\{(G,x): \forall u\in V(G), x(u)=(\cF(u),x'(u)) \;\mbox{and}\; (G,x')\notin \cF\}
\]
We show that $\MISS$ is among the hardest problems, under local label-preserving reductions. Note that $\MISS\notin\NLD$ (it is not closed under lift: it may be the case that $(G,x')\notin\cF$ but a lift of $(G,x')$ is in $\cF$). 

\begin{proposition}\label{prop:MISSisPi2-complete}\em
$\MISS$ is $\Pi_2^{\scriptsize \local}$-complete for local label-preserving reductions. 
\end{proposition}

\begin{proof}
Let $\cL$ be a distributed language. We describe a local label-preserving reduction  $(R,\cA)$ for $(\cL,\MISS)$ with respect to $\Pi_2^{\scriptsize \local}$. 

In essence, the local algorithm $\cA$ for deciding $\MISS$ in $\Pi_2^{\scriptsize \local}$ is  the generic algorithm described in the proof of Proposition~\ref{prop:Pi2isall}. Recall that, in this generic algorithm, on a legal configuration $(G,x)$, the existential $c_2$ certificate in $\cA$ is pointing to an inconsistency in the given $c_1$ certificate which is supposed to describe the configuration $(G,x)$. And, on an illegal configuration  $(G,x)$, the existential $c_1$ certificate in $\cA$ does provide an accurate description of the  configuration $(G,x)$. For the purpose of label-preservation, we slightly modify the generic algorithm for $\MISS$. Instead of viewing $c_1$ as a description of the configuration $(G,x)$, the algorithm views it as a description of $(G,x')$ where, at each node $u$, $x'(u)$ is the second item in $x(u)$ (the first item is the family $\cF(u)$). The algorithm is then exactly the same as the generic algorithm with the only modification that the test when the flag $f(u)=1$ is not regarding whether $(G,x')\in\MISS$, but whether $(G,x')\notin \cF(u)$. On a legal configuration, all nodes accept. On an illegal instance, a node with $(G,x')\in \cF(u)$ rejects. 

The reduction $R$ from $\cL$ to $\MISS$ proceeds as follows, in a way similar to the one in~\cite{FraKorPel13}. A node $u$ with identity $\id(u)$ and input $x(u)$ computes its \emph{width} $\omega(u)=2^{|\id(u)|+|x(u)|}$ where $|s|$ denotes the length of a bit-string~$s$. Then $u$ generates all configurations $(H,y)\notin \cL$ such that $H$ has at most $\omega(u)$ nodes and $y(v)$ has value at most $\omega(u)$, for every node~$v$ of $H$. It  places all these configurations in $\cF(u)$. The input $x'(u)$ is simply $x'(u)=x(u)$.  If $(G,x)\in \cL$, then $(G,x)\notin\cF$ since only illegal instances are in $\cF$, and thus $(G,R(G,x))\in \MISS$. Conversely, if $(G,x)\notin \cL$, then  $(G,R(G,x))\notin\MISS$. Indeed, there exists at least one node $u$ with identity $\id(u)\geq n$, which guarantees that $u$  generates the graph $G$. If no other node $u'$ has width $\omega(u')>n$ then $u$ generates $(G,x)\in\cF(u)$. If there exists a node $u'$ with $\omega(u')>n$ then $u'$ generates $(G,x)\in\cF(u')$. In each case, we have $(G,x)\in\cF$, and thus $(G,R(G,x))\notin\MISS$.

It remains to show that the  existential certificate used in $\cA$ for all configurations $(G,R(G,x))$ are the same for any given $(G,x)$, independently of the identity assignment to $G$ used to perform the reduction $R$. This directly follows from the nature of $\cA$ since the certificates do not depend on the families $\cF(u)$'s but only on the bit strings $x'(u)$'s. 
\end{proof}

The following language is defined as $\MISS$ by replacing $\cF$ by the closure under lift $\cF^{\uparrow}$ of~$\cF$. That is, $\cF^{\uparrow}$ is composed of $\cF$ and all the lifts of the configurations in $\cF$. 
\[
\MISSLIFT=\{(G,x): \forall u\in V(G), x(u)=(\cF(u),x'(u)) \;\mbox{and}\; (G,x')\notin \cF^{\uparrow}\}
\]
We show that $\MISSLIFT$ is the hardest problem in $\NLD$. 

\begin{proposition}\label{prop:MISSLIFTISCOMPLETE}\em
$\MISSLIFT$ is $\NLD$-complete (and $\overline{\MISSLIFT}$ is {\rm co-}$\NLD$-complete) under  label-preserving reduction. 
\end{proposition}

\begin{proof}
We do have $\MISSLIFT\in \NLD$ because $\MISSLIFT$ is closed under lift. Let $\cL\in\NLD$. The reduction $R$ from $\cL$ to $\MISSLIFT$ is the same as the one in the proof of Proposition~\ref{prop:MISSisPi2-complete}. We describe a local  algorithm for deciding $\MISSLIFT$ in $\NLD$ which is label-preserving with respect to $R$.  The certificate $c(u)$ is a description of $(G,x')$ of the form $(M(u),\data(u),\ind(u))$ as certificate~$c_1$ in the proof of Proposition~\ref{prop:Pi2isall}. This guarantees label-preservation with respect to $R$. For the verification part, each node $u$ checks whether $(M(u),\data(u),\ind(u))$  fits with its local neighborhood. If no, it rejects. Otherwise, it checks whether $(M(u),\data(u))\notin \cF(u)^{\uparrow}$,  accepts if yes, and rejects otherwise. On a legal configuration $(G,x)\in\MISSLIFT$, with the correct certificate $c=(G,x')$, all nodes accept. On an illegal configuration $(G,x)\notin\MISSLIFT$, there are two cases. If $(M(u),\data(u))$ is neither $(G,x')$ nor a lift of $(G,x')$, then some node will detect an inconsistency, and reject. If  $(M(u),\data(u))$ is an accurate description of $(G,x')$ or of a lift of $(G,x')$, then some node will detect that $(M(u),\data(u))\in\cF^{\uparrow}(u)$, and therefore will reject. 
\end{proof}

\section{Conclusion}
  
Our investigation raises several intriguing questions. In particular, identifying a distributed language in $\Pi_1^{\scriptsize \local}\setminus \LD$ was an uneasy task. We succeeded to find one such language, but we were unable to identify a $\Pi_1^{\scriptsize \local}$-complete problem, if any. In fact, completeness results are very sensitive to the type of local reductions that is used. We have identified label-preserving local reduction as an appropriate notion. It would be interesting to know whether $\NLD$-complete and $\Pi_2^{\scriptsize \local}$-complete languages exist for \emph{identity-oblivious} reductions. This latter type of reductions is indeed the most natural one in a context in which nodes may not want to leak information about their identities. It is easy to see that the class $\coLD$ has  a complete language for identity-oblivious reductions, namely, $\OR$ is $\coLD$-complete for identity-oblivious reductions. However, we do not know whether this can be achieved for $\NLD$ or~$\Pi_2^{\scriptsize \local}$. 

This paper is aiming at providing a proof of concept for the notion of interactive local verification: $\Pi_2^{\scriptsize \local}$ can be viewed as the interaction between two players, with conflicting objectives, one is aiming at proving the instance, while the other is aiming at disproving it. As a consequence, for this first attempt, we voluntarily ignored important parameters such as the size of the certificates, and the individual computation time, and we focussed only on the locality issue. The impact of limiting the certificate size was recently investigated in~\cite{FFH16}. Regarding the individual computation time,  our completeness results involve local reductions that are very much time consuming at each node. Insisting on local reductions involving polynomial-time computation at each node is crucial for practical purpose. At this point, we do not know whether non trivial hardness results can be established under polynomial-time local reductions. Proving or disproving the existence of such hardness results is let as an open problem.

\paragraph{Acknowledgement:} The authors are thankful to Laurent Feuilloley for fruitful discussions about the topic of the paper.

\newpage

\bibliographystyle{plain}


\end{document}